\documentclass{llncs}
\usepackage[english]{babel}
\usepackage[utf8]{inputenc}
\usepackage{amssymb,amsmath}
\usepackage{graphicx}
\usepackage{color}
\usepackage{xspace}
\usepackage{tikz}

\usepackage{hyperref}
\newtheorem{observation}[theorem]{Observation}
\newenvironment{keywords}{	
       \list{}{\advance\topsep by0.35cm\relax\small
       \leftmargin=1cm
       \labelwidth=0.35cm
       \listparindent=0.35cm
       \itemindent\listparindent
       \rightmargin\leftmargin}\item[\hskip\labelsep
                                     \bfseries Keywords:]}
     {\endlist}

\renewenvironment{proof}[1][Proof]{\begin{trivlist}
\item[\hskip \labelsep {\bfseries #1}]}{\hfill \qed \end{trivlist}}

\renewcommand{\problem}[3]{\begin{tabular}{lp{9cm}} PROBLEM: & {\sc #1}\\ {\it Instance}: & #2\\ Goal: & #3\end{tabular}}

\newcommand{\dproblem}[3]{\begin{tabular}{lp{9cm}} PROBLEM: & {\sc #1}\\ {\it Instance}: & #2\\ Question: & #3\end{tabular}}

\newcommand{\paramproblem}[4]{\begin{tabular}{lp{9cm}} PROBLEM: & {\sc #1}\\ {\it Instance}: & #2\\ {\it Parameter}: & #3\\ Goal: & #4\end{tabular}}

\newcommand\disjcup{\mathbin{\dot{\cup}}}

\newcommand{\N}{\mathbb{N}}

\newcommand{\NP}{{\mathsf{NP}}}

\newcommand{\coNPpoly}{{\mathsf{coNP/poly}}}
\newcommand{\FPT}{\mathsf{FPT}}
\newcommand{\W}{\mathsf{W}}
\newcommand{\XP}{\mathsf{XP}}
\newcommand{\pw}{\operatorname{pw}}
\newcommand{\tw}{\operatorname{tw}}
\newcommand{\td}{\operatorname{td}}

\newcommand{\ba}{\mathbf{a}}
\newcommand{\bb}{\mathbf{b}}

\newcommand{\MLBC}{{\sc Minimum Length Bounded Cut}\xspace}
\newcommand{\GMLBC}{{\sc Gapped Minimum Length Bounded Cut}\xspace}

\begin{document}

\title{Parameterized complexity of length-bounded cuts and multi-cuts \thanks{Research was supported by the project SVV-2014-260103.}}

\author{
Pavel Dvořák and
Dušan Knop\thanks{Author supported by the project Kontakt LH12095, project GAUK 1784214 and project CE-ITI P202/12/G061}
}

\institute{
Department of Applied Mathematics, Charles University, Prague {\tt \{koblich,knop\}@kam.mff.cuni.cz}
}

\maketitle

\begin{abstract}
We show that the {\sc Minimal Length-Bounded $L$-But} problem can be computed in linear time with respect to $L$ and the tree-width of the input graph as parameters.
In this problem the task is to find a set of edges of a graph such that after removal of this set, the shortest path between two prescribed vertices is at least $L$ long.
We derive an $\FPT$ algorithm for a more general multi-commodity length bounded cut problem when parameterized by the number of terminals also.

For the former problem we show a $\W[1]$-hardness result when the parameterization is done by the path-width only (instead of the tree-width) and that this problem does not admit polynomial kernel when parameterized by tree-width and $L.$
We also derive an $\FPT$ algorithm for the {\sc Minimal Length-Bounded Cut} problem when parameterized by the tree-depth. Thus showing an interesting paradigm for this problem and parameters tree-depth and path-width.
\end{abstract}

\begin{keywords}
length bounded cuts, parameterized algorithms, $\W[1]$-hardness
\end{keywords}

\section{Introduction}

The study of network flows and cuts begun in 1960s by the work of Ford and Fulkerson~\cite{ford-fulkerson}. It has many generalizations and applications now. We are interested in a generalization of cuts related to the flows using only short paths.

\paragraph{Length bounded cuts} Let $s,t \in V$ be two distinct vertices of a graph $G=(V,E)$ -- we call them source and sink, respectively. We call a subset of edges $F\subseteq E$ of $G$ an {\sc $L$-bounded cut} (or {\sc $L$-cut} for short), if the length of the shortest path between $s$ and $t$ in the graph $(V, E\setminus F)$ is at least $L+1$. We measure the length of the path by the number of its edges. In particular, we do not require $s$ and $t$ to be in distinct connected components as in the standard cut, instead we do not allow $s$ and $t$ to be close to each other. We call the set $F$ {\it a minimum $L$-cut} if it has the minimum size among all $L$-bounded cuts of the graph $G$.

%Note that there is also a dual problem called {\sc Length Bounded Flow}. For this we are interested in a maximal flow between vertices $s$ and $t$ that can be decomposed into paths of length at most $L$. As this can be computed in polynomial time by use of linear programming, we are not interested in this problem in this paper. For details see Kolman and Scheideler~\cite{kolman2}.

We state the cut problem formally:

%  \vskip 0.2cm
  \problem{\MLBC (MLBC)}
	{graph $G=(V,E)$, vertices $s,t$ and integer $L\in\N$}
	{find a minimum $L$-bounded $s,t$ cut $F\subset E$}
%  \vskip 0.5cm

Length bounded flows were first considered by Adámek and Koubek~\cite{koubek-adamek}. They showed that the max-flow min-cut duality cannot hold and also that integral capacities do not imply integral flow.
Finding a minimum length bounded cut is $\NP$-hard on general graphs for $L \ge 4$ as was shown by Itai et al.~\cite{itai}. They also found algorithms for finding $L$-bounded cut with $L=1,2,3$ in polynomial time by reducing it to the usual network cut in an altered graph. The algorithm of Itai et al.~\cite{itai} uses the fact that paths of length $1,2$ and $3$ are edge disjoint from longer paths, while this does not hold for length at least $4$.

Baier et al.~\cite{kolman} studied linear programming relaxation and approximation of {\sc MLBC} together with inapproximability results for length bounded cuts. They also showed instances of the {\sc MLBC} having $O(L)$ integrality gap for their linear programming approach, which are series-parallel graphs and thus have constant bounded tree-width.
The first parametrized complexity study of this and similar topics was made by Golovach and Thilikos~\cite{golovach} who studied parametrization by paths-length (that is in our setting the parameter $L$) and the size of the solution for cuts. They also proved hardness results -- finding disjoint paths in graphs of bounded tree-width is a $\W [1]$-hard problem.

The {\sc MLBC} problem has its applications in the network design and in the telecommunications. Huygens et al.~\cite{pesneau} use a {\sc MLBC} as a subroutine in the design of $2$-edge-connected networks with cycles at most $L$ long. The {\sc MLBC} problem is called {\it hop constrained} in the telecommunications and the number $L$ is so called number of hops. The main interest is in the constant number of hops, see for example the article of Dahl and Gouveia~\cite{dahl}.

Note that the standard use of the Courcelle theorem~\cite{courcelle} gives for each fixed $L$ a linear time algorithm for the decision version of the problem. But there is no apparent way of changing these algorithms into a single linear time algorithm. Moreover there is a nontrivial dependency between the formula (and thus the parameter $L$) and the running time of the algorithm given by Courcelle theorem.

Now we give a formal definition of a rather new graph parameter, for which we give one of our results:
\begin{definition}[Tree depth~\cite{NdM:tree-depth}]
The closure $Clos(F)$ of a forest $F$ is the graph obtained from $F$
by making every vertex adjacent to all of its ancestors. The tree-depth $\td(G)$ of
a graph $G$ is one more than the minimum height of a rooted forest $F$ such that
$G\subseteq Clos(F).$ 
\end{definition}

\paragraph{Our Contribution}
Our main contribution is an algorithm for the {\sc MLBC} problem, its consequences and an algorithm for a more general multi-terminal version problem. 

\begin{theorem}\label{thm:main_theorem}
Let $G$ be a graph of tree-width $k$. Let $s$ and $t$ be two distinct vertices of $G$. Then for any $L\in\N$ an minimum $L$-cut between $s$ and $t$ can be found in time $O((L^{k^2})^2\cdot 2^{k^2}\cdot n)$.
\end{theorem}

\begin{corollary}\label{thm:vertexCover}
Let $G$ be a graph, $k = \td(G)$ and $s$ and $t$ be two distinct vertices of $G$. Then for any $L\in\N$ an minimum $L$-cut between $s$ and $t$ can be found in time $f(k)n$, where $f$ is a computable function.
\end{corollary}

\begin{proof}
As $k$ is the tree depth of $G$ it follows that the length of any path in $G$ can be upper-bounded by $2k$. It is a folklore fact, that $k$ is also a upper-bound on the tree-width of $G$. So we can use the Theorem~\ref{thm:main_theorem} for $L < 2k$ and any other polynomial algorithm for minimum cut problem otherwise.
\end{proof}

\begin{corollary}\label{thm:xp-corollary}
Let $G = (V, E)$ be a graph of tree-width $k$, $s \neq t \in V$ and $L \in \N$. There exists a computable function $f:\N\rightarrow\N$, such that a minimum $L$-cut between $s$ and $t$ can be found in time $O(n^{f(k)})$, where $n = |V|$.
\end{corollary}

\begin{theorem}\label{thm:hardness}
{\sc Minimal Length Bounded Cut} parametrized by path-width is $\W[1]$-hard.
\end{theorem}

\paragraph{Tree-width versus tree-depth}
Admitting an $\FPT$ algorithm for a problem when parameterized by the tree-width implies an $\FPT$ algorithm for the problem when parameterized by the tree-depth, as parameter-theoretic observation easily shows. On the other hand, the $\FPT$ algorithm parameterized by the tree-width usually uses exponential (in the tree-width) space, while the tree-depth version uses only polynomial space (in the tree-depth).

From this point of view, it seems to be interesting to find problems that are "on the edge between path-width and tree-depth". That is problems that admit an $\FPT$ algorithm when parameterized by the tree-depth, but being $\W$[1]-hard when parameterized by the path-width.

The only other result of this type, we a re aware of in the time of writing this article is by Gutin et al.~\cite{GJW:MixedChinese}. The {\sc Minimum Length Bounded Cut} problem is also a problem of this kind---as Theorems~\ref{thm:hardness} and \ref{thm:vertexCover} demonstrate.

Theorem~\ref{thm:main_theorem} gives us that the {\sc MLBC} problem is fixed parameter tractable ($\FPT$) when parametrized by the length of paths and the tree-width and that it belongs to $\XP$ when parametrized by the tree-width only (and thus solvable in polynomial time for graph classes with constant bounded tree-width).

\begin{theorem}\label{thm:refuteKernel}
There is no polynomial kernel for the \MLBC problem parameterized by the tree-width of the graph and the length $L,$ unless $\NP\subseteq\coNPpoly.$
\end{theorem}

We want to mention that our techniques apply also for more general version of the {\sc MLBC} problem.

\paragraph{Length-bounded multi-cut}
We consider a generalized problem, where instead of only two terminals, we are given a set of terminals. For every pair of terminals, we are given a constraint---a lower bound on the length of the shortest path between these terminals. More formally:

Let $S = \{s_1,\dots,s_k\} \subset V$ be a subset of vertices of the graph $G=(V,E)$ and let $a:S\times S \rightarrow \N$ be a mapping. We call a subset of edges $F\subseteq E$ of $G$ an {\it $\ba$-bounded multi-cut} if length of the shortest path between $s_i$ and $s_j$ in the graph $(V, E\setminus F)$ is at least $a(s_i,s_j)$ long for every $i \neq j$. Again if $F$ has smallest possible size, we call it {\it minimum $\ba$-bounded $\{s_1,\dots,s_k\}$-multi-cut}. We call the vertices $s_1,\dots, s_k$ {\it terminals}. Finally, as there are only finitely many values of the mapping $a$ we write $a_{s,t}$ instead of $a(s,t)$, we also write $\ba$ instead of function $a$. Let $L \ge \max_{s,t\in S} a(s,t)$, we say that the problem is {\it $L$-limited}.

%  \vskip 0.5cm
  \problem{Minimum Length Bounded Multi-Cut (MLBMC)}	%% neriakt tomu radeji Triangle Minim...??
	{graph $G=(V,E)$, set $S\subset V$ and $a_{s,t}\in\N$ for all $s,t\in S,$ satisfying the triangle inequalities}
	{find a minimum length bounded $S$ multi-cut $F\subset E$}
% \vskip 0.5cm

\begin{theorem}\label{thm:multiterminal}
Let $G=(V,E)$ be a graph of tree-width $k$, $S\subseteq V$ with $|S| = t$ and let $p:= t+k$.
Then for any $L\in\N$ and any $L$-limited length-constraints $\ba$ on $S$ an minimum $\ba$-bounded multi-cut can be computed in time $O((L^{p})^2\cdot 2^{p^2}\cdot n)$.
\end{theorem}

%Note that Theorem~\ref{thm:multiterminal} does not rely on the triangle inequalities.

\section{Preliminaries}\label{sec:preliminaries}
%\paragraph{Notation} %znaèení
In this section we recall some standard definitions from the graph theory and state what a tree decomposition is. 
After this we introduce changes of the tree decomposition specific for our algorithm. 
We proceed by the notion of auxiliary graphs used in proofs of our algorithm correctness.
Finally, in Section~\ref{s:kernelRefutePrelim} we summarize the results allowing us to prove that it is unlikely for a parameterized problem to admit a polynomial kernelization procedure.

%\paragraph{}
%The {\it degree} of a vertex $v$ in a graph $(V,E)$ is $deg(v) = |\{e\in E\colon v\in e\}|$ (the number of neighbors of $v$). Recall that a graph is {\it connected}, if there exists a path between every pair of its vertices. A {\it tree} is a connected acyclic graph, a {\it star} is a tree that has at most one vertex of degree greater than~1. We call a vertex {\it leaf}, if it has degree 1. For a graph $G$ and $U\subseteq V(G)$ we denote by $G[U] := (U,E\cap\binom{U}{2})$ the {\it induced subgraph of the graph} $G$ (in this case it is induced by vertices of the set $U$). For a set $S$ we denote by $2^S$ the set of all subsets of $S$.

We use the notion of tree decomposition of the graph:

\begin{definition}
A {\em tree decomposition} of a graph $G=(V,E)$ is a pair ${\cal T} = (\{B_i\colon i\in I\},T = (I,F)),$ where $T$ is a rooted tree and $\{B_i\colon i\in I\}$ is a family of subsets of $V,$ such that
\begin{enumerate}
  \item for each $v\in V$ there exists an $i \in I$ such that $v\in B_i$,
  \item for each $e\in E$ there exists an $i \in I$ such that $e\subseteq B_i$,
  \item for each $v\in V, I_v = \{i \in I\colon v\in B_i\}$ induces a subtree of $T.$
\end{enumerate}
We call the elements $B_i$ the {\it nodes}, and the elements of the set $F$ the decomposition edges.
\end{definition}
We define a width of a tree decomposition ${\cal T} = (\{B_i\colon i\in I\}, T)$ as $\max_{i\in I}|B_i|-1$ and the {\it tree-width} $\tw(G)$ of a graph $G$ as the minimum width of a tree decomposition of the graph $G$. 
Moreover, if the decomposition is a path we speak about the {\it path-width} of $G$, which we denote as $\pw(G)$. 

\paragraph{Nice tree decomposition} \cite{kloks}~
For algorithmic purposes it is common to define a {\it nice tree decomposition} of the graph. 
We naturally orient the decomposition edges towards the root and for an oriented decomposition edge $(B_j,B_i)$ from $B_j$ to $B_i$ we call $B_i$ the {\it parent} of $B_j$ and $B_j$ a {\it child} of $B_i$.
If there is an oriented path from $B_j$ to $B_i$ we say that $B_j$ is a {\it descendant} of $B_i$.

We also adjust a tree decomposition such that for each decomposition edge $(B_i,B_j)$ it holds that $| |B_i|-|B_j| | \le 1$ (i.e. it joins nodes that differ in at most one vertex). The in-degree of each node is at most $2$ and if the in-degree of the node $B_k$ is $2$ then for its children $B_i,B_j$ holds that $B_i = B_j = B_k$ (i.e. they represent the same vertex set).

We classify the nodes of a nice decomposition into four classes---namely {\it introduce nodes}, {\it forget nodes}, {\it join nodes} and {\it leaf nodes}. We call the node $B_i$ an introduce node of the vertex $v$, if it has a single child $B_j$ and $B_i\setminus B_j = \{v\}$. We call the node $B_i$ a forget node of the vertex $v$, if it has a single child $B_j$ and $B_j\setminus B_i = \{v\}$. If the node $B_k$ has two children, we call it a join node (of nodes $B_i$ and $B_j$). Finally we call a node $B_i$ a leaf node, if it has no child.

\begin{proposition} \cite{kloks} \label{thm:nodes_prop}  %% vznik dekompozice
Given a tree decomposition of a graph $G$ with $n$ vertices that has width $k$ and $O(n)$ nodes, we can find a nice tree decomposition of $G$ that also has width $k$ and $O(n)$ nodes in time $O(n)$.
\end{proposition}

So far we have described a standard nice tree decomposition.
Now we change the introduce nodes. 
Let $B_j$ be an introduce node and $B_i$ its parent.
We add another two copies $B_j^p, B_j^s$ of $B_j$ to the decomposition. 
We remove decomposition edge $(B_j,B_i)$ and add three decomposition edges $(B_j,B_J^p), (B_j^s, B_j^p)$ and $(B_j^p, B_i)$.
Note that after this operation, $B_j^s$ is a leaf of the decomposition, $B_j$ remains an introduce node and $B_j^p$ is a join node. We call $B_j^s$ a sibling of $B_j$.
% (see Fig.~\ref{fig:decomposition-introduction}). 
Note that by these further modifications we preserve linear number of nodes in the decomposition.

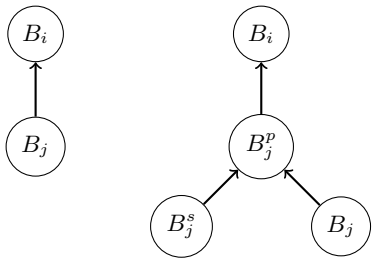
\begin{figure}[htb!]
\centering
\begin{tikzpicture}[node distance = 1.5cm]
  \node[draw,circle] (BBi) {$B_i$};
  \node[below of=BBi,draw,circle] (BBj) {$B_j$};
  \path[->,draw,thick] (BBj) -- (BBi);

  \node[draw,circle] at (3,0) (Bi) {$B_i$};
  \node[below of=Bi,draw,circle] (Bjp) {$B_j^p$};
  \node[below left of=Bjp,draw,circle] (Bjs) {$B_j^s$};
  \node[below right of=Bjp,draw,circle] (Bj) {$B_j$};

  \path[->,draw,thick] (Bjp) -- (Bi);
  \path[->,draw,thick] (Bj) -- (Bjp);
  \path[->,draw,thick] (Bjs) -- (Bjp);
  
\end{tikzpicture}
\caption{Change of Introduction nodes in the nice tree decomposition}
\label{fig:decomposition-introduction}
\end{figure}

\paragraph{Auxiliary subgraphs}
Recall that for each edge there is at least one node containing that particular edge. %We choose for each edge a node, moreover it is possible to choose a leaf in our decomposition.
Note that after our modification of the decomposition for each edge $e$ there is at least one leaf $B_i$ of the decomposition satisfying $e \subseteq B_i$. 
To see this, suppose this is not true and that some edge, say $e$, must be placed into a non-leaf node $B_j$. 
We may suppose that $B_j$ is an introduce node (for join or forget node choose its descendant).
However, in our construction any introduce node $B_j$ has a sibling $B_j^s$ such that is a leaf in the decomposition tree and their bags are equal.

Thus, for every edge $e \in E(G)$ we choose an arbitrary leaf node $B_i$ such that $e\in B_i$ and say that the edge $e$ {\it belongs} to the leaf node $B_i$. 
By this process we have chosen set $E_i\subset E(G)$ for each leaf node $B_i$.
We further use the notion of {\it auxiliary graph} $G_{B_i},$ or $G_i$ for short.
For a leaf node $B_i$ we set a graph $G_i = (V_i = B_i, E_i)$.
For a non-leaf node $B_i$ we set a graph $G_i = (V,E)$, where $V = B_i\cup \bigcup_{B_j \text{ child of } B_i} V_j$ and $E = \bigcup_{B_j \text{ child of } B_i} E_j$.

%\todo{predlat a uhladit -- tvorba grafu}
%For a node $X$ of the graph we construct graph $G_X$ as follows. If $X$ is a leaf of the decomposition we set $V(G_X) := X$ and $E(G_X)$ to the set of the edges belonging to leaf $X$. If $X$ is not a leaf of the decomposition we set $G_X$ to be the union of graphs of its descendants, formally $G_X := \cup_{Y {\text descendant } X} G_Y$.

%Note that the generalized {\sc MLBMC} version of the problem is harder than the {\sc MLBC} version, not only by containing it as a special case. Let us demonstrate this on a very small example -- let us form an instance of the problem by graph $G$, terminals $\{s,t,u\}$ and length constraints such that $l_{s,t} > l_{t,u} + l_{u,s}$ holds (e.g. $l_{s,t} = 20$, $l_{t,u} = 3$ and $l_{s,u} = 4$).
%It is obvious that for fulfilling the length constraint for $s$ and $t$, one must also add some artificial length to the constraints for $t$ and $u$ and for $s$ and $u$, respectively, such that there is no shortcut path between $s$ and $t$ via the vertex $u$. Let us mark there lengths as $b_{t,u}$ and $b_{u,s}$, respectively.
%The right solution is the also characterized by $\bb$. We say that the solution for the instance of the problem is {\it consistent}, if there are no shortcuts (as described in previous example).

\subsection{Preliminaries on refuting polynomial kernels}\label{s:kernelRefutePrelim}
Here we present simplified review of a framework used to refute existence of polynomial kernel for a parameterized problem from Chapter 15 of a monograph by Cygan et al.~\cite{CFKLMPPS-FPT}.

In the following we denote by $\Sigma$ a final alphabet, by $\Sigma^*$ we denote the set of all words over $\Sigma$ and by $\Sigma^{\le n}$ we denote the set of all words over $\Sigma$ and length at most $n.$

\begin{definition}[Polynomial equivalence relation]\label{d:polyRelation}
An equivalence relation $\mathcal{R}$ on the set $\Sigma^*$ is called {\em polynomial equivalence relation} if the following conditions are satisfied:
\begin{enumerate}
  \item There exists an algorithm such that, given strings $x,y\in\Sigma^*,$ resolves whether $x\equiv_\mathcal{R} y$ in time polynomial in $|x| + |y|$.
  \item Relation $\mathcal{R}$ restricted to the set $\Sigma^{\le n}$ has at most $p(n)$ equivalence classes for some polynomial $p(\cdot).$
\end{enumerate}
\end{definition}

\begin{definition}[Cross-composition]\label{d:crossComposition}
Let $L\subseteq\Sigma^*$ be an unparameterized language and $Q\subseteq\Sigma^*\times\N$ be a parametrized language. We say that $L$ {\em cross-composes} into $Q$ if there exists a polynomial equivalence relation $\mathcal{R}$ and an algorithm $\mathcal{A},$ called the cross-composition, satisfying the following conditions. 
The algorithm $\mathcal{A}$ takes on input a sequence of strings $x_1,x_2,\dots,x_t\in\Sigma^*$ that are equivalent with respect to $\mathcal{R}$, runs in polynomial time in $\sum_{i = 1}^t |x_i|,$ and outputs one instance $(y,k)\in\Sigma^*\times\N$ such that:
\begin{enumerate}
  \item $k\le p(max_{i = 1}^t |x_i|, \log t)$ for some polynomial $p(\cdot,\cdot),$ and
  \item $(y,k)\in Q$ if and only if $x_i\in L$ for all $i.$
\end{enumerate}
\end{definition}

With this framework, it is possible to refute even stronger reduction techniques---namely polynomial compression---according to the following definition:

\begin{definition}[Polynomial compression]\label{d:polyCompression}
A {\em polynomial compression} of a parameterized language $Q\subseteq\Sigma^*\times\N$ into an unparameterized language $R\subseteq\Sigma^*$ is an algorithm that takes as an input an instance $(y,k)\in\Sigma^*\times\N,$ works in polynomial time in $|x| + k,$ and returns a string $y$ such that:
\begin{enumerate}
  \item $|y|\le p(k)$ for some polynomial $p(\cdot),$ and
  \item $y\in R$ if and only if $(x,k) \in Q.$
\end{enumerate}
\end{definition}

It is possible to refute existence of polynomial kernel using Definitions~\ref{d:polyRelation},\ref{d:crossComposition} and \ref{d:polyCompression} with the help of use of the following theorems and a complexity assumption that is unlikely to hold---namely $\NP\subseteq\coNPpoly$.

\begin{theorem}[\cite{Drucker12}]
Let $L,R\subseteq\Sigma^*$ be two languages. Assume that there exists an AND-distillation of $L$ into $R.$ Then $L\in\coNPpoly.$
\end{theorem}

\begin{theorem}
Assume that an $\NP$-hard language $L$ AND-cross-composes to a parameterized language $Q.$ Then $Q$ does not admit a polynomial compression, unless $\NP\subseteq\coNPpoly.$
\end{theorem}

\section{Minimal Length Bounded Multi-cuts}
In this section we give a more detailed study of the length constraints for the length-bounded multi-cut and the triangle inequalities. From this we derive Lemma~\ref{thm:edge-disjoint-MLBMC} for merging solutions for edge-disjoint graphs.

\paragraph{The triangle inequalities}
Note that the solution for {\sc MLBMC} problem has to satisfy the triangle inequalities with respect to its instance. 
This means that for any three terminals $s,t,u\in S$ and the distance function $dist$ it holds that $dist(s,u) + dist(u,t) \ge dist(s,t) \ge a_{s,t}$.
Thus, we can restrict instances of {\sc MLBMC} problem only to those satisfying these triangle inequalities.

\begin{definition}[Length constraints]
Let $G = (V,E)$ be a graph, $S\subset V$ and let $k = |S|$. We call a vector $\ba = (a_{s_1, s_2},\dots,a_{s_{k - 1}, s_k})$ a length constraint if for every $s,t,u \in S$ it holds that $a_{s,u} + a_{u,t} \ge a_{s,t}$.
\end{definition}

For our approach it is important to see the structure of the solution on a graph composed from two edge disjoint graphs. 
%The proof of the following lemma can be found in the Appendix.

\begin{lemma}\label{thm:edge-disjoint-MLBMC}
Let $G_1 = (V_1,E_1), G_2 = (V_2,E_2)$ be edge disjoint graphs.
Then for the graph $G = G_1\cup G_2$ and $S = V_1\cap V_2$ and an arbitrary length constraints $\ba \in\N^{\binom{|S|}{2}}$ it holds that	%% satisfying triangle
the minimum length bounded $(S,\ba)$ multi-cut $F$ for $G$ is a disjoint union of the $(S,\ba)$ multi-cuts $F_1$ and $F_2$ for $G_1$ and $G_2$.
\end{lemma}

\begin{proof}
First we prove that there cannot be smaller solution than $F_1\cup F_2$.
To see this observe that for every $(S,\ba)$ cut $F'$ on $G$ it holds that $F'\cap E_1$ is an $(S,\ba)$ cut on $G_1$ (and vice versa for $G_2$).
Hence if $F'$ would be a cut of smaller size then $F$, we would get a contradiction with the minimality of choice of $F_1$ and $F_2$, because we would have $|F'| < |F| = |F_1| + |F_2|$.

Now we prove that $F_1\cup F_2$ is a valid solution.
To see this we prove that every path between two terminals is not short. We prove that there cannot be such $P$ by an induction on number of $h := |V(P) \cap S|$.
If $h = 2$ then because $G_1$ and $G_2$ are edge disjoint, we may (by symmetry) assume that $P\subset G_1$, a contradiction with the choice of $F_1$.
If $h > 2$ then there is a vertex $u\in S\setminus\{s,t\}$ such that the path $P$ is composed from two segments $P_1$ and $P_2$, where $P_1$ is a path between $s$ and $u$ and $P_2$ is a path between $u$ and $t$. And so we have $|P| = |P_1| + |P_2| \ge a_{s,u} + a_{u,t} \ge a_{s,t}$, what was to be demonstrated.
\end{proof}

\section{Restricted bounded multi-cut}

In this section we present our approach to the $L$-bounded cut for the graphs of bounded tree-width. We present our algorithm together with some remarks on our results.

Recall that we use dynamic programming techniques on a tree decom\-position of an input graph. 
First we want to root the decomposition in a node containing both source and sink of the $L$-cut problem. 
This can be achieved by adding the source to all nodes on the unique path in the decomposition tree between any node containing the source and any node containing the sink. 
Note that this may add at most $1$ to the width of the decomposition.

\paragraph{Length vectors and tables}
As it was mentioned in the previous section, we solve the $L$-cut by reducing it to simple instances of generalized {\sc MLBMC} problem. We begin with a mapping $a: S\times S \rightarrow \N$ with meaning $a(s,t) = l_{s,t}$.
For simplicity we represent the mapping $a$ by a vector, calling it a {\it length vector} $\ba$ and relax it for a node $X$ $\ba = (a_{x_1,x_2},\dots,a_{x_{k-1},x_k}) \in \N^{\binom{k}{2}}$, where $k = |X|$ and $X = \{x_1,\dots,x_k\}$. 
We reduce the problem to the $\ba$-bounded multi-cut for $k$ terminals, where $k = tw(G) + 2$ (the additional two is for changing the decomposition). 
Let us introduce a relation on length vectors $\ba, \bb \in \N^{\binom{k}{2}}$ on $X$ of the same size. We write $\ba \preceq \bb$, if $a_{x_i,x_j} \le b_{x_i,x_j}$ for all $1\le i<j\le |X|$.

Let the set of vertices $X = \{x_1,\dots, x_k\}$, $\ba$ be a length vector, let $I\subset [k]$ and let $Y = \{x_i\in X\colon i\in I\}$.  
By $\ba|_Y$ we denote the length vector $\ba$ containing $a_{x_i,x_j}$ if and only if both $i\in I$ and $j\in I$ (in an appropriate order) -- in this case we say $\ba|_Y$ is $\ba$ {\it contracted} on the set $Y$.

Recall that for each node $X$ we have defined the auxiliary graph $G_X$ (see Section~\ref{sec:preliminaries} for the definition). With a node $X$ we associate the table $Tab_X$. The table entry for constraints $\ba = (a_{1,2},\dots,a_{k-1,k})$ of $Tab_X$ (denoted by $Tab_X[\ba]$) for the node $X = \{x_1,\dots,x_k\}$ contains the size of the $\ba$-bounded multi-cut for the set $X$ in the graph $G_X.$ Note that for two length vectors $\ba\preceq\bb$ it holds that $Tab_X[\ba] \le Tab_X[\bb]$.

%\paragraph{Compatible Pairs of Length vectors}
%We use the {\sc MLBMC} problem for merging solutions for smaller graphs to get a solution for a larger graph. Let $G_1 = (V_1,E_1), G_2 = (V_2,E_2)$ and $G = (V,E)$ be graphs with $V = V_1\cup V_2, E = E_1\cup E_2$ and $E_1\cap E_2 = \emptyset$. Let $X = V_1\cap V_2$. And let $(G, X, \ba)$ be an instance of {\sc MLBMC} we say that pair $(\bb,\bc)$ of length vectors in compatible with $\ba$ if additional triangle inequalities $\ba_{s,t} \le \bb_{s,u} + \bc_{u,t}$ are fulfilled for each $s,t,u \in X$.

%\todo{dodelat obrazek zkratek a napsat neco ke zkratkam a forgetum?}

\subsection{Node Lemmas}

The leaf nodes are the only nodes bearing some edges. We use an exhaustive search procedure for building tables for these nodes. For this we need to compute the lengths of the shortest paths between all the vertices of the leaf node, for which we use the well known procedure due to Floyd and Warshall~\cite{floyd,warshall}:

\begin{proposition}[\cite{floyd,warshall}]\label{thm:floyd-warshall}
Let $G$ be a graph with nonnegative length $f: G(E) \rightarrow \N$. It is possible to compute the table of lengths of the shortest paths between any pair $u, v \in V(G)$ with respect to $f$ in time $O(|V(G)|^3)$.
\end{proposition}

%We use the Propositon~\ref{thm:floyd-warshall} to prove the following lemma (the proof can be found in Appendix).

\begin{lemma}[Leaf Nodes]\label{thm:leafnodes}
For all $L$-limited length vectors and a leaf node $X$ the table $Tab_X$ of sizes of minimum length-bounded multi-cuts can be computed in time $O(L^{k^2}\cdot 2^{k^2}\cdot k^3)$, where $k = |X|$.
\end{lemma}

\begin{proof}
Fix one $L$-limited length vector $\ba$. We run the Floyd-Warshall algorithm (stated as Proposition~\ref{thm:floyd-warshall}) for every possible subgraph of $G_X$. As $|E(G_X)| \le \binom{k}{2}$ there are $O(2^{k^2})$ such subgraphs. This gives us a running time $O(2^{k^2}\cdot k^3)$ for a single length vector $\ba$ -- we check all $O(k^2)$ length constraints if each of them is satisfied and set $$Tab_X[\ba] := \min_{F\subseteq E(G_X)\colon F \text{ satisfies } \ba} (|E(G_X) \setminus F|).$$

Finally there are $O(L^{k^2})$ $L$-bounded length vectors, this gives our result.
\end{proof}

%\begin{proof}
%Fix one $L$-limited length vector $\ba$. We run the Floyd-Warshall algorithm (stated as Proposition~\ref{thm:floyd-warshall}) for every possible subgraph of $G_X$. As $|E(G_X)| \le \binom{k}{2}$ there are $O(2^{k^2})$ such subgraphs. This gives us a running time $O(2^{k^2}\cdot k^3)$ for a single length vector $\ba$ -- we check all $O(k^2)$ length constraints if each of them is satisfied and set $$Tab_X[\ba] := \min_{F\subseteq E(G_X)\colon F \text{ satisfies } \ba} (|E(G_X) \setminus F|).$$

%Finally there are $O(L^{k^2})$ $L$-bounded length vectors, this gives our result.
%\end{proof}

We now use Lemma~\ref{thm:edge-disjoint-MLBMC} to prove time complexity of finding a dynamic programming table for join nodes from the table of its children.

%\paragraph{Proof remarks}
%Let us remark -- it is not obvious from the first reading of the proof that we do not throw away some clever solution. For example there may be two solutions of the very same value which are combinatorially distinct (do not cut the same edges). We want to note that if we have a power to recognize these from each other -- they are close in a meaning choosing one partial solution may cause some vertices to be further apart then while it is not so when choosing the other solution. This cannot be because if there are distinct solutions of same size they either assure same length of paths between every pair of terminals (then we cannot distinguish them bud for the solution they look the same) or there are pair where the partial solutions assure distinct length of paths between some of pairs so typically one of the solutions became compatible while the other do not.

\begin{lemma}[Join Nodes]\label{thm:join_lemma}
Let $X$ be a join node with children $Y$ and $Z$, let $L$ be the limit on length vectors components and let $k = |X|$. Then the table $Tab_X$ can be computed in time $O(L^{k^2})$ from the table $Tab_Y$ and $Tab_Z$.
\end{lemma}
\begin{proof}
%Recall that graphs $G_Y$ and $G_Z$ assigned to nodes $Y$ and $Z$ are edge disjoint by their definition. Fix $L$-limited length vector $\ba$. Note that we cannot assign $Tab_X[\bl] := Tab_Y[\bl] + Tab_Z[\bl]$ as for the reason that this can produce a shortcut! The path can be split into two (edge disjoint) subpaths through some vertex $u\in X$ -- one subpath in the graph $G_Y$ and the other in the graph $G_Z$.
Recall that graphs $G_Y$ and $G_Z$ are edge disjoint and that we store sizes of $\ba$-bounded multi-cuts. Note also that $X = V(G_Y)\cap V(G_Z)$ and so we can apply Lemma~\ref{thm:edge-disjoint-MLBMC} and set $Tab_X[\ba] := Tab_Y[\ba] + Tab_Z[\ba],$ for each $\ba$ satisfying the triangle inequalities. 
%So we check all $L$-limited length vectors $\bb, \bc \succeq\ba$, if for all $x,y,z\in X$ (distinct) $b_{x,z} + c_{z,y} \ge a_{x,y}$ and $c_{x,z} + b_{z,y} \ge a_{x,y}$, we mark the pair $(\bb,\bc)$ as compatible (with vector $\ba$), we mark the pair as non-compatible otherwise. Finally we set $$Tab_X[\ba] = \min_{(\bb,\bc)\text{ compatible pair}}(Tab_Y[\bb] + Tab_Z[\bc]).$$

%For the complexity issue, there are at most $(L^{k^2})$ of (all) $L$-limited length vectors. We need to check $O(k^3)$ length constraints for each $\ba$.
As there are $O(L^{k^2})$ entries in the table $Tab_X$ we have the complexity we wanted to prove.
\end{proof}

As the forget node expects only forgetting a vertex and thus forgetting part of the table of the child. This is the optimizing part of our algorithm. 

\begin{lemma}[Forget Nodes]\label{thm:forget_lemma}
Let $X$ be a forget node, $Y$ its child, let $L$ be the limit on length vectors components and let $k = |X|$. Then the table $Tab_X$ can be computed in time $O((L^{k^2})^2)$ from the table $Tab_Y$.
\end{lemma}
\begin{proof}
Fix one length vector $\ba$ and compute the set ${\mathcal A}(\ba)$ of all $Y$-augmented length vectors. 
Formally $\bb \in {\mathcal A}(\ba)$ if $\bb$ is a length vector satisfying the triangle inequalities for $Y$ and $\bb|_X = \ba$. 
After this we set 
\[
Tab_X[\ba] := \min_{\bb\in {\mathcal A}(\ba)} Tab_Y[\bb]. 
\]

There are at most $L^{k^2}$ of $Y$-augmented length vectors for each $\ba$ and this gives the claimed time.
\end{proof}

Also the introduce node (as the counter part for the forget node) only adds coordinates to the table of its child. It does no computation as there are no edges it can decide about -- these nodes now only add isolated vertex in the graph.

\begin{lemma}[Introduce Nodes]\label{thm:introduce_lemma}
Let $X$ be an introduce node, $Y$ its child, let $L$ be the limit on length vectors components and let $k = |X|$. Then the table $Tab_X$ can be computed in time $O(L^k)$ from the table $Tab_Y$.
\end{lemma}
\begin{proof}
Let $x$ be the vertex with the property $x \in X\setminus Y$. The key property is that the vertex $x$ is an isolated vertex in $G_X$ and thus we can set $Tab_X[\ba] := Tab_Y[\ba|_Y],$ because $x$ is arbitrarily far from any vertex in $G_Y$, especially from the set $Y$.
\end{proof}

\subsection{Proofs of Theorems}
We use Lemmas~\ref{thm:leafnodes},~\ref{thm:join_lemma},~\ref{thm:forget_lemma} and~\ref{thm:introduce_lemma} to prove the theorem about computing $L$-bounded $(s,t)$-cut in graph of bounded tree-width. 
For this, note that we can use $k = O(tw(G))$ and put it into all the Lemmas as it is an upper bound on the size of any node in the decomposition.% (the $O$ notation is because $+1$ in the definition and possible $+2$ for producing of node with both $s$ and $t$ in it).

We compute all the $L$-bounded length-constraint that satisfy the triangle inequalities in advance. This takes additional time $O(L^{k^2}\cdot k^3)$ which can be upper-bounded by $O((L^{k^2})^2\cdot 2^k)$ for $k \ge 2$ and $L\ge 2$ and so this does not make the overall time complexity worse.

\begin{proof}[Proof of Theorem~\ref{thm:main_theorem}]
As there are $O(n)$ nodes in nice tree decomposition (by Proposition~\ref{thm:nodes_prop}) and as we can upper-bound time needed to compute any type of node by $O((L^{k^2})^2\cdot 2^{k^2})$, we have complexity proposed in state of the Theorem~\ref{thm:main_theorem}.
\end{proof}

Let us now point out that the value of the parameter $L$ can be upper-bounded by the number of vertices $n$ of the input graph $G$ (in fact by $n^{1-\varepsilon}$ as it is proved in~\cite{kolman}).

%So the Corollary~\ref{thm:xp-corollary} gives us an $\XP$ algorithm and thus saying (as the Theorem~\ref{thm:main_theorem}) that for all the graphs with bounded tree-width, the {\sc MLBC} problem is solvable in polynomial time. However, if more of the structure is used one can obtain much better polynomial time -- our approach gives an $O(n^{13})$ time algorithm for series-parallel graph, while there is an $O(n^7)$ algorithm (as proven in~\cite{knop-diplo}).

We now want to sum-up the key ideas leading to Theorem~\ref{thm:main_theorem}. First, it is the use of dynamic program for computing all options of cuts for bounded number of possible choices and for second it is the idea of creating a node that includes both the source and the sink while not harming the tree-width too much. On the other hand, we can use this idea to solve also the generalized version of the problem -- length-bounded multi-cut -- with the additional parameter the number of terminals. It is easy to see that in this setting that again it is possible to achieve node containing every terminal and thus this yields following Theorem~\ref{thm:triangle-multiterminal}.

\begin{theorem}\label{thm:triangle-multiterminal}
Let $G=(V,E)$ be a graph of tree-width $k$, $S\subseteq V$ with $|S| = t$ and let $p:= t+k$.
Then for any $L\in\N$ and any $L$-limited length-constraints $\bb$ satisfying the triangle inequalities on $S$ an minimum $\bb$-bounded multi-cut can be computed in time $O((L^{p})^2\cdot 2^{p^2}\cdot n)$.
\end{theorem}

\begin{proof}
To see that Theorem~\ref{thm:triangle-multiterminal} implies Theorem~\ref{thm:multiterminal} note, that given any $L$-limited length-constraints $\ba$ the minimum $\ba$-bounded multi-cut must satisfy triangle inequalities. 
Thus, we find the minimum $\bb$-bounded multi-cut among all $\bb \succeq \ba$. Note that we can do this in additional time $O(L^{p^2})$, but this does not make the total running time worse.

\end{proof}
%To see that Theorem~\ref{thm:triangle-multiterminal} implies Theorem~\ref{thm:multiterminal} note, that given any $L$-limited length-constraints $\ba$ the minimum $\ba$-bounded multi-cut must satisfy triangle inequalities. Thus we find the minimum $\bb$-bounded multi-cut among all $\bb \succeq \ba$. Note that we can do this in additional time $O(L^{p^2})$, but this does not make the total running time worse.

\section{Hardness of the $L$-bounded cut}
\label{sec:hardness}
In this section we prove MLBC parametrized by path-width is $\W[1]$-hard~\cite{Flum} by $\FPT$-reduction from $k$-{\sc Multicolor Clique}.
%\vskip 0.1cm

  \paramproblem{$k$-{\sc Multicolor Clique}}
	{$k$-partite graph $G=(V_1 \disjcup V_2 \disjcup \dots \disjcup V_k,E)$, where $V_i$ is independent set for every $i$ and they are pairwise disjoint}
	{$k$}
	{find a clique of the size $k$}
%\vskip 0.5cm

\paragraph{Denoting} In this section, sets $V_1, \dots, V_k$ are always partites of the $k$-partite graph $G$.  
We denote edges between $V_i$ and $V_j$ by $E_{ij}$. 
The problem is $\W[1]$-hard even if every independent set $V_i$ has the same size and the number of edges between every $V_i$ and $V_j$ is the same. 
In whole Section~\ref{sec:hardness} we denote the size of an arbitrary $V_i$ by $N$ and the size of an arbitrary $E_{ij}$ by $M$. 
For an $\FPT$-reduction from $k$-{\sc Multicolor Clique} to MLBC we need:
\begin{enumerate}
\item Create an MLBC instance $G' = (V', E'), s, t, L$ from the $k$-{\sc Multicolor Clique} instance $G = (V_1 \disjcup V_2 \disjcup \dots \disjcup V_k,E)$ where the size of $G'$ is polynomial in the size of $G$.
\item Prove that $G$ contains a $k$-clique if and only if $G'$ contains an $L$-bounded cut of the size $f(k,N,M)$ where $f$ is a polynomial.
\item Prove the path-width of $H$ is smaller than $g(k)$ where $g$ is a computable function.
\end{enumerate}

Our ideas were inspired by work of Michael Dom et al.~\cite{cvc-hardness}. They proved $\W[1]$ hardness of {\sc Capacitated Vertex Cover} and {\sc Capacitated Dominating Set} parametrized by the tree width of the input graph. 
We remarked that their reduction also proves $\W[1]$ hardness of these problems parametrized by path-width. 

\subsection{Basic gadget}
In the $k$-{\sc Multicolor Clique} problem we need to select exactly one vertex from each independent set $V_i$ and exactly one edge from each $E_{ij}$. 
Moreover, we have to make certain that if $e \in E_{ij}$ is the selected edge and $u \in V_i, v \in V_j$ are the selected vertices then $e = \{u,v\}$. 
The idea of the reduction is to have a basic gadget for every vertex and edge. 
We connect gadgets $g_v$ for every $v$ in $V_i$ into a path $P_i$. 
The path $P_i$ is cut in the gadget $g_v$ if and only if the vertex $v \in V_i$ is selected into the clique. 
The same idea will be used for selecting the edges. 

\begin{definition}
Let $h, Q \in \N$. Butte $B(s',t', h, Q)$ is a graph which contains $h$ paths of length $2$ and $Q$ paths of length $h + 2$ between the vertices $s'$ and $t'$. 
The short paths (of length $2$) are called shortcuts, the long paths are called ridgeways and the parameter $h$ is called height.
\end{definition}
A butte for $h = 3, Q = 4$ is shown in Figure~\ref{fig:butte} part a.
In our reduction all buttes will have the same parameter $Q$ (it will be computed later). 
For simplicity we depict buttes as a dash dotted line triangles with its height $h$ inside (see Figure~\ref{fig:butte} part b), or only as triangles without the height if it is not important. 

\begin{figure}
\begin{center}
\includegraphics[scale=1]{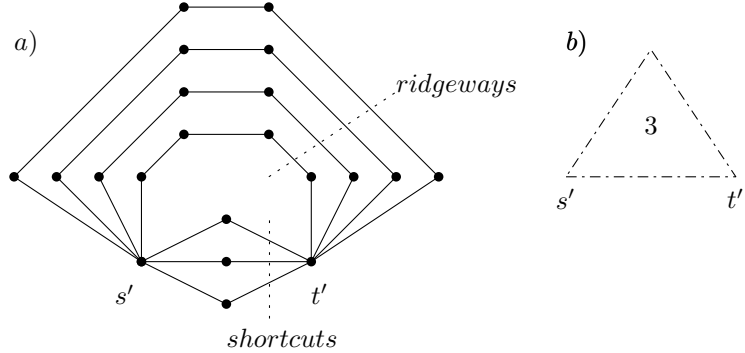}
\caption{a) Example of a butte for $h = 3$ and $Q = 4$. b) Simply diagram for a butte of height $3$. }
\label{fig:butte}
\end{center}
\end{figure}

Let $B(s',t',h,Q)$ be a butte. 
We denote by $s(B), t(B), h(B), Q(B)$ the parameters of butte $B$ $s',t',h$ and $Q$, respectively. 
We state easy but important observation about the butte path-width:% (the proof is in Appendix):

\begin{observation}
\label{obs:pathwidth}
Path-width of an arbitrary butte $B$ is at most 3.
\end{observation}

\begin{proof}
If we remove vertices $s(B)$ and $t(B)$ from $B$ we get $Q(B)$ paths from ridgeways and $h(B)$ isolated vertices from shortcuts. 
This graph has certainly path-width 1. 
If we add $s(B)$ and $t(B)$ to every node of the path decomposition we get a proper path decomposition of $B$ with width 3.
\end{proof}

Let butte $B(s',t',h,Q)$ be a subgraph of a graph $G$. Let $u,v$ be vertices of $G$ and all paths between $u$ and $v$ goes through $B$ such that they enter into $B$ in $s'$ and leave it in $t'$ (see Figure~\ref{fig:pathbutte}). 
The important properties of the butte $B$ are:
\begin{enumerate}
\item By removing one edge from all $h$ shortcuts of $B$, we extend the the distance between $u$ and $v$ by $h$. If we cut all shortcuts of butte $B$ we say the butte $B$ is ridged. 
\item Let size of the cut is bounded by $K \in \N$ and we can remove edges only from $B$. If we increase $Q$ to be bigger than $K$ then any path between $u$ and $v$ cannot be cut by removing edges from $B$ (only extended by ridging the butte $B$).
%\item Butte $B$ has constant path-width independent on $h$ and $Q$.
\end{enumerate}

\begin{figure}
\begin{center}
\includegraphics[scale=1]{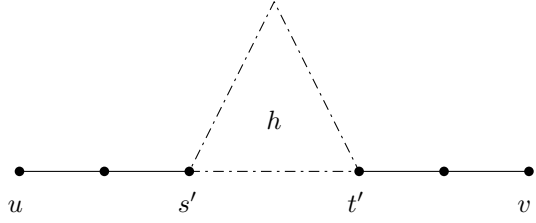}
\caption{Example of a path going through a butte.}
\label{fig:pathbutte}
\end{center}
\end{figure}

\subsection{Butte path}
In this section we define how we connect buttes into a path, which we call highland. The main idea is to have highland for every pair $(i, j), i \neq j \in [k]$. 
In the highland for $(i, j)$, there are buttes for every vertex $v \in V_i$ and every edge $e \in E_{i,j}$. We connect vertex buttes and edge buttes into a path. Then we set the butte heights and limit the size of the cut in such way that:
\begin{enumerate}
\item Exactly one vertex butte and exactly one edge butte have to be ridged.
\item If a butte for a vertex $v$ is ridged, then only buttes for edges incident with $v$ can be ridged.
\end{enumerate}
Formal description of the highland is in the following definition.
\begin{definition}
The highland $H(X, Y, s, t)$ is a graph containing 2 vertices $s$ and $t$ and $Z = X + Y$ buttes $B_1, \dots , B_Z$ where:
\begin{enumerate}
\item $s = s(B_1), t = t(B_Z)$ and $t(B_i) = s(B_{i + 1})$ for every $1 \leq i < Z$.
\item $h(B_i) = X^2 + i$ for $1 \leq i \leq X$.
\item $h(B_i) \in \{X^4, \dots, X^4 + X - 1\}$ for $X + 1 \leq i \leq Z$.
\item $Q(B_i) = X^4 + X^2$ for every $i$.
\end{enumerate}
\end{definition}
Let $H(X,Y,s,t)$ be a highland. We call buttes $B_1, \dots, B_X$ from $H$ low and buttes $B_{X + 1}, \dots, B_{X + Y}$ high (low buttes will be used for the vertices and high buttes for the edges). 
The vertex $t(B_X) = s(B_{X+1})$, where low and high buttes meet, is called the center of highland $H$. Note that there can be more buttes with the same height among high buttes and they are not ordered by height as the low buttes. 
An example of a highland is shown in Figure~\ref{fig:highland}.

\begin{figure}
\begin{center}
\includegraphics[scale=1]{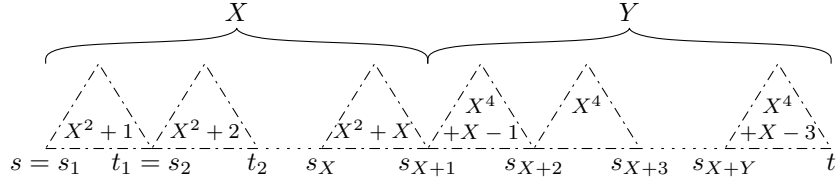}
\caption{Example of a highland $H(X,Y,s,t)$.}
\label{fig:highland}
\end{center}
\end{figure}

\begin{proposition}
\label{prp:highland}
Let $H(X,Y,s,t)$ be a highland. Let $L = 2(X + Y) + X^4 + X^2 + X - 1$. Let $C$ be the $L$-cut of size $X^4 + X^2 + X$, which cut all paths of length $L$ and shorter between $s$ and $t$ then:
\begin{enumerate}
\item The cut $C$ contains only edges obtained by ridging the exactly two buttes $B_i, B_j$, such that $B_i$ is low and $B_j$ is high.
\item Let $B_i$ be the ridged low butte and $B_{j}$ be the ridged high butte. Then, $h(B_j) = X^4 + X - i$.
\end{enumerate} 
\end{proposition}

\begin{proof}[Proof of Proposition~\ref{prp:highland}]
Every butte has at least $X + 1$ shortcuts and $X^4 + X^2$ ridgeways. Therefore, $C$ can not cut all paths in $H$ between $s$ and $t$ and it is useless to add edges from ridgeways to the cut $C$. Note that the shortest $st$-path in $H$ has the length $2(X + Y)$. 
\begin{enumerate}
\item If we ridge every low butte we extend the shortest $st$-path by $X^3 + \frac{X^2}{2} + \frac{X}{2}$. However, it is not enough and at least one high butte has to be ridged. Two high buttes cannot be ridged otherwise the cut would be bigger then the bound. No high butte can extend the shortest $st$-path enough, therefore at least one low butte has to be ridged. However, two low buttes and one high butte cannot be ridged because the cut $C$ would be bigger then the bound.
\item Let $F$ be the set of removed edges from ridged buttes $B_i$ and $B_j$. The height of $B_i$ is $X^2 + i$. Therefore, the length of the shortest $st$-path after ridging $B_i$ and $B_j$ and the size of $F$ is $2(X + Y) + X^2 + i + h(B_j)$. If $h(B_j) < X^4 + X - i$ then shortest $st$-path is strictly shorter then $2(X + Y) + X^4 + X^2 + X$. Thus, $F$ is not $L$-cut. If $h(B_j) > X^4 + X - i$ then and $|F| > X^4 + X^2 + X$ thus $F$ is bigger than $C$.
\end{enumerate}
\end{proof}

\subsection{Reduction}
In this section we present our reduction. Let $G=(V_1 \disjcup V_2 \disjcup \dots \disjcup V_k,E)$ be the input for $k$-{\sc Multicolor Clique}. As we stated in the last section, the main idea is to have a low butte $B_v$ for every vertex $v \in V(G)$ and a high butte $B_e$ for every edge $e \in E(G)$. Vertex $v$ and edge $e$ is selected into the $k$-clique if and only if the butte $B_v$ and the butte $B_e$ are ridged. 
From $G$ we construct {\sc MLBC} input $G', s, t, L$ (the construction is quite technical, for better understanding see Figure~\ref{fig:reduction}):
\begin{enumerate}
\item For every $1 \leq i,j \leq k, i \neq j$ we create highland $H^{i,j}(N,M,s,t)$ of buttes $B^{i,j}_1, \dots, B^{i,j}_{N + M}$.
\item Let $V_i = \{v_1, \dots, v_N\}$. 
The vertex $v_\ell \in V_i$ is represented by the low butte $B^{i,j}_\ell$ of the highland $H^{i,j}$ for every $j \neq i$. 
Thus, we have $k - 1$ copies of buttes (in different highlands) for every vertex. Hence, we need to be certain that only buttes representing the same vertex are ridged. Note that buttes representing the same vertex have the same height and the same distance from the vertex $s$.
\item Let $E_{ij} = \{ e_1, \dots, e_M\}, i < j$. 
Edge $e_\ell = \{u, v\} \in E_{ij} (u \in V_i, v \in V_j)$ is represented by the high butte $B^{i,j}_{N + \ell}$ of the highland $H^{i,j}$ and by the high butte $B^{j,i}_{N + \ell}$ of the highland $H^{j,i}$. Note that two buttes represented the same edge has same distance from the vertex $s$. 
Let $h_i, h_j$ be the heights of buttes representing the vertices $u$ and $v$, respectively. 
We set the buttes heights:
\begin{enumerate}
\item $h(B^{i,j}_{N + \ell}) = N^4 + N - h_i$
\item $h(B^{j,i}_{N + \ell}) = N^4 + N - h_j$
\end{enumerate}
\item We add edge $\bigl\{t(B^{i,j}_\ell), t(B^{i, j + 1}_\ell)\bigr\}$ for every $1 \leq i \leq k, 1 \leq j < k, i \neq j$ and $1 \leq \ell < N$.
\item We add paths of length $N - 1$ connected  $t(B^{i,j}_\ell)$ and $t(B^{j, i}_\ell)$ for every $1 \leq i,j \leq k, i \neq j$ and $1 \leq \ell < M$.
\item We set $L$ to $2(N + M) + N^4 + N^2 + N - 1$.
\end{enumerate}
We call paths between highlands in Items 4 and 5 the valley paths.

\begin{figure}
\begin{center}
\includegraphics[scale=1]{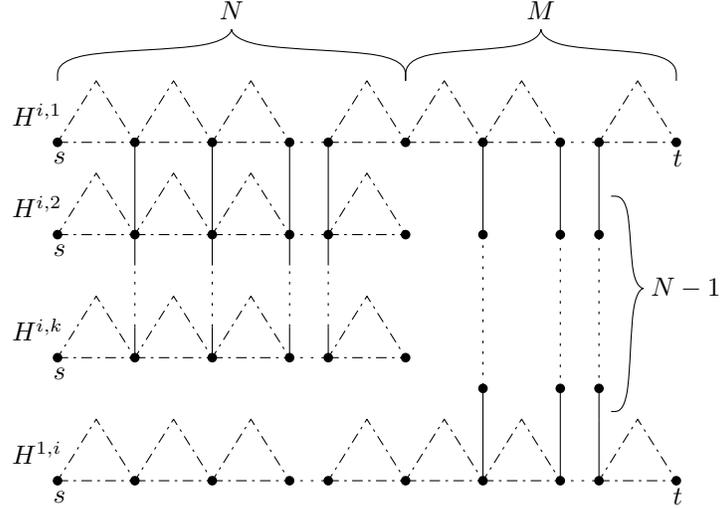}
\caption{Some part of the graph $G'$. All vertices labeled $s$ and $t$ are actually two vertices $s$ and $t$ in the graph $G'$. We divided them for better illustration. Highlands $H^{i,2}$ and $H^{i,k}$ have also high buttes, but we omitted them.}
\label{fig:reduction}
\end{center}
\end{figure}

\begin{observation}
\label{obs:polynom}
Graph $G'$ has a polynomial size in the size of the graph $G$.
\end{observation}

\begin{theorem}
\label{thm:correctness1}
If graph $G$ has a clique of size $k$ then $(G', s, t)$ has an $L$-cut of size $k(k - 1)(N^4 + N^2 + N)$.
\end{theorem}
\begin{proof}
Let $G$ has a $k$-clique $\{v_1,\dots,v_k\}$ where $v_i \in V_i$ for every $i$ and $e_{ij} = \{v_i, v_j\} \in E_{ij}$. For every $i$ we ridge all $k-1$ buttes representing the vertex $v_i$ in $G'$. And for every $i < j$ we ridge both buttes representing the edge $e_{ij}$. 

We claim that set of removed edges from ridged buttes forms the $L$-cut. Let $H^{i,j}$ be an arbitrary highland. There is no $st$-path shorter than $L$ in $H^{i,j}$. Let $h(B_v) = N^2 + \ell$ where $B_v$ is arbitrary butte representing the vertex $v_i$. By construction of $G'$, the high butte representing the edge $e_{ij}$ in $H^{i,j}$ has height $N^4 + N - \ell$. Thus, ridged buttes in $H^{i,j}$ extend the shortest $st$-path by $N^4 + N^2 + N$ and it has length $2(M + N) + N^4 + N^2 + N$. Buttes representing the vertex $v_i$ have same height. Thus, a path through the low buttes of highlands using some valley path is always longer than path going through low buttes of only one highland. Therefore, it is useless to use valley paths among low buttes for the shortest $st$-path. 

Other situation is among high buttes because buttes representing the same edge have different heights. The butte $B_v$ representing the vertex $v_i$ extend the shortest path at least by $N^2 + 1$. The butte $B_e$ representing the edge $e_{i,j}$ extend the shortest at least by $N^4$. However, if $h(B_v) + h(B_e) < N^4 + N^2 + N$ then $B_v$ and $B_e$ have to be in different highlands. Therefore, the $st$-path going through $B_v$ and $B_e$ has to use a valley path between high buttes, which has length $N - 1$. Hence, any $st$-path has the length at least $2(N + M) + N^4 + N^2 + N$. 

We remove $N^4 + N^2 + N$ edges from each highland and there are $k(k-1)$ highlands in $G'$. Therefore, $G'$ has $L$-cut of the size $k(k-1)(N^4 + N^2 + N)$.
\end{proof}

\begin{theorem}
\label{thm:correctness2}
If $(G', s, t)$ has an $L$-cut of size $k(k-1)(N^4 + N^2 + N)$ then $G$ has a clique of size $k$.
\end{theorem}
\begin{proof}
Let $C$ be an $L$-cut of $G'$. Every shortest $st$-path going through every highland has to be extended by $N^4 + N^2 + N$. By Proposition~\ref{prp:highland} (Item 1), exactly one low butte and exactly one high butte of each highland has to be ridged. We remove $(N^4 + N^2 + N)$ from every highland in $G'$. Therefore, there can be only edges from ridged buttes in $C$. 

For fixed $i$, highlands $H^{i,j}$ are the highlands which low buttes represent vertices from $V_i$. We claim that ridged low buttes of $H^{i,1}, \dots, H^{i,k}$ represent the same vertex. Suppose for contradiction, there exists two low ridged buttes $B_\ell$ of $H^{i,\ell}$ and $B_m$ of $H^{i,m}$ which represent different vertex from $V_i$. Without loss of generality $H^{i,\ell}$ and $H^{i,m}$ are next to each other (i.e. $|\ell - m| = 1$) and the distance from $s$ to $s(B_\ell)$ is smaller than the distance from $s$ to $s(B_m)$. 
Let $B'_\ell$ be a butte of $H^{i,m}$ such that it has the same distance from $s$ as the butte $B_\ell$ (see Figure~\ref{fig:shortcutbutte}).
The path $s$--$t(B'_\ell)$--$t(B_\ell)$--$t$ does not go through any ridged low butte. Therefore, this path is shorter than $L$, which is contradiction. We can use the same argument to show that there are not two high ridged buttes of highland $H^{i,j}$ and $H^{j,i}$ which represent different edges from $E_{ij}$.

\begin{figure}
\begin{center}
\includegraphics[scale=1]{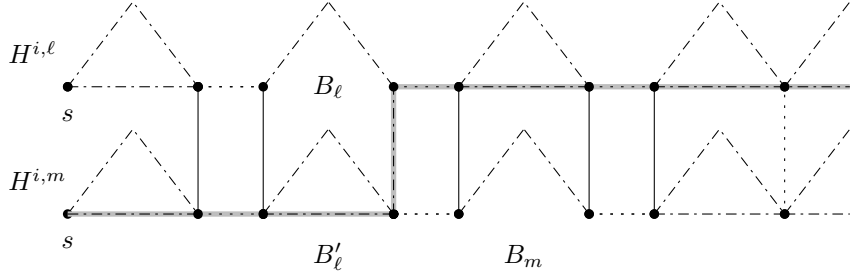}
\caption{How to miss every ridged low butte if there are ridged two low buttes representing two different vertices from one color class. Ridged butte is depicted as triangle without hypotenuse.}
\label{fig:shortcutbutte}
\end{center}
\end{figure}

We put into the $k$-clique $K \subset V(G)$ the vertex $v_i \in V_i$ if and only if an arbitrary butte representing the vertex $v_i$ is ridged. We proved in the previous paragraph that exactly one vertex from $V_i$ can be put into the clique $K$. Let $e_{ij} \in E_{ij}$ be an edge represented by ridged high buttes. We claim that $v_i \in e_{ij}$. Let $B \in H^{i,j}$ be a butte representing $v_i$ with height $N^2 + \ell$. Then by Proposition~\ref{prp:highland} (Item 2), butte $B' \in H^{i,j}$ of height $N^4 + N - \ell$ has to be ridged. By construction of $G'$, only buttes representing edges incident with $v_i$ have such height. Therefore, chosen edges are incident with chosen vertices and they form the $k$-clique of the graph $G$.
\end{proof}

\begin{observation}
\label{obs:parameter}
Graph $G'$ has path-width in $O(k^2)$.
\end{observation}

\begin{proof}
Let $H$ be a graph created from $G$ by replacing every butte by a single edge and contract the valley paths between high buttes into single edges, see Figure~\ref{fig:pathwidth} transformation $a$.  Let $U$ be a vertex set containing $s$, $t$ and every highland center. Let $H'$ be a graph created from $H$ by removing all vertex from $U$, see Figure~\ref{fig:pathwidth} transformation $b$. 

\begin{figure}
\begin{center}
\includegraphics[scale=1]{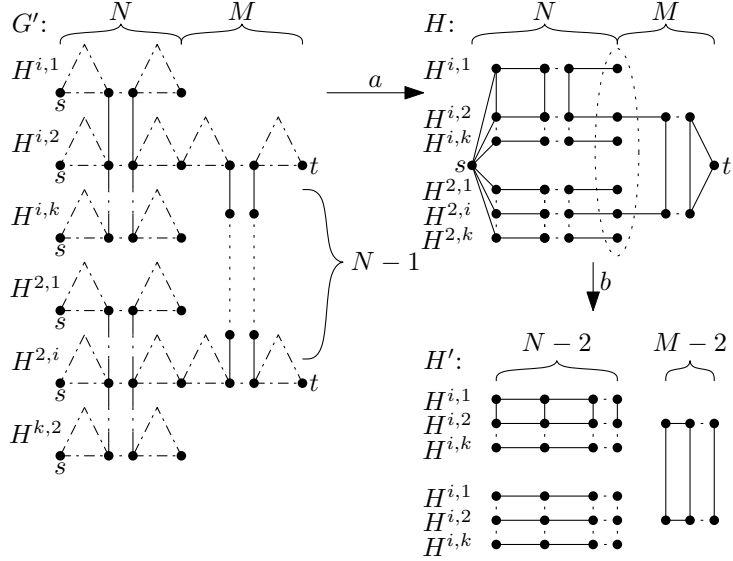}
\caption{The transformation $a$ replace all buttes in $G'$ by single edges and contract long valley paths into single edges. The transformation $b$ removes vertices $s$ and $t$ and all highland centers (highlighted by dotted ellipsis) from $H$. }
\label{fig:pathwidth}
\end{center}
\end{figure}

Graph $H'$ is unconnected and it contains $k$ grids of size $(k-1) \times (N - 2)$ and $\binom{k}{2}$ grids of size $2 \times (M - 2)$. 
Path-width of $(k - 1) \times (N - 2)$ grids is in $O(k)$, therefore $\pw(H') \in O(k)$. 
If we add set $U$ to every node of a path decomposition of $H'$ we get proper path decomposition of $H$. Since $|U| \in O(k^2)$, path-width of $H$ is in $O(k^2)$.
The edge subdivision does not increase path-width. Moreover, replacing edges by buttes does not increase it either (up to multiplication constant) because butte has the constant path-width (Observation~\ref{obs:pathwidth}). Therefore, $\pw(G) = c\pw(H)$ for some constant $c$ and $\pw(G) \in O(k^2)$. 

\end{proof}

And thus Theorem~\ref{thm:hardness} easily follows from Observations~\ref{obs:polynom} and~\ref{obs:parameter} and Theorems~\ref{thm:correctness1} and~\ref{thm:correctness2}.

\section{Polynomial kernel is questionable}
In this section, we prove that for the \MLBC problem it is unlikely to admit a polynomial kernel when parametrized by the length $L$ and the path-width (tree-width) of the input graph. We will prove this fact by the use of a AND-composition framework---that is by designing an AND-composition algorithm from the unparameterized gapped version of the \MLBC problem into itself.

\dproblem{Decision version of \MLBC}
{Graph $G = (V,E),$ vertices $s_1,s_2,$ positive integers $L,K$}
{Is there an $L$ length bounded cut consisting of exactly $K$ edges}

Baier et al.~\cite{kolman} proved that this problem is $\NP$-complete even if we restrict the input instances such that
\begin{itemize}
  \item the desired length $L \geq 4$ is constant,
  \item either there is an $L$ bounded cut of size exactly $K$ or every $L$ bounded cut has at least $1.13K$ edges.
\end{itemize}
Thus, it is possible to define the polynomial relation $\mathcal{R}$ as follows. We will consider instances $(G,s_1,s_2,L,K)$ of the decision problem with constant $L.$ Two instances $(G,s_1,s_2, L, K), (G',s_1',s_2',L,K')$ are equivalent if $|V(G)| = |V(G')|$ and $K = K'.$ It is clear that $\mathcal{R}$ is a polynomial equivalence relation.

\paragraph{AND-composition}
We will take graphs $G_1,G_2,\dots,G_t$ the input instances that are equivalent under the relation $\mathcal{R}$ for the \GMLBC problem with the constant $L= 4.$ As a composition we will take the disjoint union of graphs $G_1,G_2,\dots,G_t$ and unify sources and sinks of the resulting graph and denote the graph as $G.$
 
Now it is easy to see that the path-width of the graph $G$ of our construction is at most $\max_{i = 1,2,\dots, t} |G_i| + 2$ as we can form bags of the path decomposition by whole graphs $G_i$ add the source and the sink into every bag and connect them into a path. This is indeed a correct path decomposition as the only two vertices in common are in every bag.

As we have taken the parameter $L$ to be constant, this finishes arguments about parameters. And thus also the proof of Theorem~\ref{thm:refuteKernel}.

\section{Conclusions}

There is another standard generalization of the length bounded cut problem -- where we add to each edge also its length. If this length is integral it is possible to extend and use our techniques (we only subdivide edges longer than $1$ -- this doesn't raise the tree-width of the graph on the input). On the other hand, if we allow fractional numbers, it is uncertain how to deal with such a generalization.

\paragraph{Acknowledgments}
Authors thank to Jiří Fiala, Petr Kolman and Lukáš Folwarczný for fruitful discussions about the problem. We would like to mention that part of this research was done during Summer REU 2014 at Rutgers University.

\bibliographystyle{siam}
\bibliography{lbounded,monographs,denseGraphs}

\end{document}